\documentclass{revtex4}
\usepackage{amsmath,amsthm}
\usepackage{rotating}
\usepackage{multirow}

\usepackage{graphicx}

\usepackage{amsfonts}
\newtheorem{theorem}{Theorem}[section]
\newtheorem{lemma}[theorem]{Lemma}
\newtheorem{proposition}[theorem]{Proposition}

\theoremstyle{remark}
\newtheorem{remark}[theorem]{Remark}

\theoremstyle{definition}
\newtheorem{definition}[theorem]{Definition}

\theoremstyle{example}
\newtheorem{example}[theorem]{Example}

\theoremstyle{notation}

\newcommand{\bra}[1]{\langle#1|}
\newcommand{\ket}[1]{|#1\rangle}

\begin{document}

\title{M\"obius operators and non-additive quantum probabilities in the Birkhoff-von Neumann lattice}            
\author{A. Vourdas}
\affiliation{Department of Computer Science,\\
University of Bradford, \\
Bradford BD7 1DP, United Kingdom\\a.vourdas@bradford.ac.uk}

\begin{abstract}
The properties of quantum probabilities are linked to the geometry of quantum mechanics, described by the Birkhoff-von Neumann lattice.
Quantum probabilities violate the additivity property of Kolmogorov probabilities, and they are interpreted as Dempster-Shafer probabilities.
Deviations from the additivity property are quantified with the M\"obius (or non-additivity) operators which are defined through M\"obius transforms, and which are shown to be intimately related to commutators. 
The lack of distributivity in the Birkhoff-von Neumann lattice $\Lambda _d$, causes
deviations from the law of the total probability (which is central in Kolmogorov's probability theory). 
Projectors which quantify the lack of distributivity in $\Lambda _d$, and also deviations from the law of the total probability, are introduced.
All these operators, are observables and they can be measured experimentally. 
Constraints for the M\"obius operators, which are based on the properties of the Birkhoff-von Neumann lattice
(which in the case of finite quantum systems is a modular lattice), are derived.
Application of this formalism in the context of coherent states, generalizes coherence to multi-dimensional structures.
\end{abstract}
\maketitle
MSC: 60A05, 06C05

\section{Introduction}
Probability theory needs to be compatible with the geometry of the structure, in which it is applied. 
In the case of quantum mechanics, the properties of quantum probabilities need to be compatible with the Birkhoff-von Neumann lattice\cite{LO1,LO2,LO3,LO4,LO5,LO6,LO7}.
Probability theories in a quantum context\cite{G1,G2,G3,G4,G5,G6,G7,G8}, include
operational approaches and convex geometry methods \cite{CO1,CO2,CO3,CO4,CO5},
test spaces \cite{TE1,TE2},
fuzzy phase spaces \cite{FF1,FF2},
category theory methods\cite{CA1,CA2}, etc.

In recent work \cite{VO1,VO2} we interpreted quantum probabilities as non-additive probabilities.
The basic property of Kolmogorov probabilities is additivity, which can be written as $q(A\cup B)-q(A)-q(B)+q(A\cap B)=0$, and which for exclusive events reduces to $q(A\cup B)-q(A)-q(B)=0$. 
Non-additive probabilities have been studied for a long time, at an abstract level, e.g. the early work at a philosophical level \cite{K1}, which was translated into mathematical axioms  in \cite{K2,K3}, and discussed further in \cite{K4} (the history of non-additive probabilities is discussed in \cite{K5}).
The basic idea is expressed in everyday language as `the whole is greater than the sum of its parts'.
For example, the percentage of votes for a coalition of political parties, might be greater than the sum of the percentages for each component party.
More recently the Dempster-Shafer theory \cite{D0,D1,D2,D3,D4,D5,D6} led to a practical scheme 
of non-additive probabilities, which has been used extensively in Artificial Intelligence, Operations Research, Game theory, Mathematical Economics, etc (its place within the general area of Bayesian probabilities is discussed in \cite{KY}).

We have shown\cite{VO1} a direct link between non-commutativity and the non-additivity of quantum probabilities, and we have 
interpreted quantum probabilities in the context of the Dempster-Shafer theory.
In particular, we have introduced the non-additivity operator ${\mathfrak D}(H_1, H_2)$  which measures deviations from the additivity of probability (see Eq.(\ref{32}) below) ,
and proved that it is related to the commutator of the projectors $\Pi(H_1), \Pi(H_2)$, to the subspaces $H_1,H_2$, as described in Eq.(\ref{e3}).

Probability theory is tacitly related to logic, because it needs the concepts of 
conjunction, disjunction, and negation.
Kolmogorov probability is intimately related to Boolean logic which is related to set theory, and is appropriate for Classical Mechanics.
Quantum Mechanics is based on
the Birkhoff-von Neumann orthomodular lattice of closed subspaces of the Hilbert space.
We consider systems with finite-dimensional Hilbert space, in which case this is a modular orthocomplemented lattice.
Within this lattice there are sublattices which are Boolean algebras, and in these `islands' quantum probability obeys the additivity property and 
it can be interpreted as Kolmogorov probability. But in the full lattice quantum probabilities violate the additivity property,
and we interpreted them as Dempster-Shafer probabilities. An important difference between Boolean algebras and modular lattices, is that the property of distributivity in the former, 
is replaced with modularity (Eq.(\ref{292}) below) which is a `weak distributivity', in the latter.

In this paper:
\begin{itemize}
\item
We use M\"obius transforms to introduce M\"obius (or non-additivity) operators that involve many subspaces, and show their relation to commutators 
(section \ref{s1}).
In this sense, the formalism of non-additive probabilities, complements the non-commutativity formalism.
\item
We quantify the lack of distributivity in the modular lattice $\Lambda _d$, with the projectors 
$\varpi _1(H_1,H_2|H_0)$ and $\varpi _2(H_1,H_2|H_0)$ in Eq.(\ref{40}). 
The lack of distributivity causes deviations from the law of the total probability (which is fundamental for Kolmogorov probabilities).
These deviations are quantified with the
projectors $\pi (H_0;H_1)$ in Eq.(\ref{123}) (section \ref{s3}). 

\item
All of these operators  are observables and they can be measured experimentally.
If quantum probabilities were Kolmogorov probabilities, they would be zero (section \ref{s4}). 
\item
We use the properties of modular lattices to find constraints for the non-additivity operators
${\mathfrak D}(H_1, H_2)$ (propositions \ref{P1},\ref{P2},\ref{P3}, in section \ref{s5}).
This links probability theory with the geometry (Birkhoff-von Neumann lattice) of quantum mechanics.
\item
We use these projectors and operators in the context of coherent states.
This generalizes coherence to multi-dimensional structures (section \ref{s10}).
\end{itemize}

\section{Preliminaries}
\subsection{Kolmogorov probabilities versus Dempster-Shafer probabilities}

Let $\Omega$ be a set of alternatives, with finite cardinality $|\Omega|$.
The powerset $2^{\Omega}$ contains $2^{|\Omega|}$ subsets of $\Omega$, for which we use the notation $A_i$ with $i=1,...,2^{|\Omega|}$.
The $2^{\Omega}$ is a Boolean algebra with intersection ($\cap$), union ($\cup$), and complement (${\overline A}=\Omega -A$),
as conjunction, disjunction, and negation, correspondingly.
Kolmogorov probability is based on this Boolean algebra, and it assigns to $A_i$ a number $p(A_i)\in [0,1]$, such that
\begin{eqnarray}\label{A}
&&p(\emptyset )=0;\;\;\;\;p(\Omega)=1\\
&&A_1\cap A_2=\emptyset\;\;\rightarrow\;\;p(A_1\cup A_2)= p(A_1)+p(A_2)\label{1}
\end{eqnarray}
The last equation is the additivity relation, which can also be written in a more general way, in order to include non-exclusive events, as:
\begin{eqnarray}\label{1}
\delta (A_1,A_2)=0;\;\;\;\;\;\delta (A_1,A_2)=p(A_1\cup A_2)-p(A_1)-p(A_2)+p(A_1\cap A_2).
\end{eqnarray}
Using this we can prove the following relations, that involve three sets:
\begin{eqnarray}\label{112}
\delta (A_1,A_2,A_3)&=&{\widetilde  \delta} (A_1,A_2,A_3)=0\nonumber\\
\delta (A_1,A_2,A_3)&=&p(A_1\cup A_2\cup A_3)-p(A_1\cup A_2)-p(A_1\cup A_3)-p(A_2\cup A_3)\nonumber\\&+&p(A_1)+p(A_2)+p(A_3)-p(A_1\cap A_2\cap A_3)\nonumber\\
{\widetilde  \delta} (A_1,A_2,A_3)&=&
p(A_1\cap A_2\cap A_3)-p(A_1\cap A_2)-p(A_1\cap A_3)-p(A_2\cap A_3)\nonumber\\&+&p(A_1)+p(A_2)+p(A_3)-p(A_1\cup A_2\cup A_3)
\end{eqnarray}
The proof is straightforward, but we stress that it uses the distributivity property of set theory.
In a quantum context later, the distributivity property is not valid.
$\delta (A_1,A_2,A_3)$ and ${\widetilde \delta}(A_1,A_2,A_3)$ are dual to each other, in the sense that $\cup$ and $\cap$ in $\delta (A_1,A_2,A_3)$,
are replaced with $\cap$ and $\cup$ in ${\widetilde \delta}(A_1,A_2,A_3)$.

We can generalize these formulas using the M\"obius transform, which has been used extensively in Combinatorics, 
after the pioneering work by Rota\cite{R1,R2}.
In its simplest form, it is used in the inclusion-exclusion principle that gives the cardinality of the union of overlapping sets.
The M\"obius transform describes the overlaps between sets, and it can be used to avoid the `double-counting'.
More generally, the M\"obius transform is applied to partially ordered structures,
and in this paper we use it with the Birkhoff-von Neumann modular lattice of subspaces $\Lambda _d$.

We introduce the M\"obius transform as
\begin{eqnarray}\label{2}
&&{\cal M}(A_1,...,A_n)=\sum _{E} (-1)^{n-|E|}\;p\left (\bigcup _{i\in E} A_i\right );\;\;\;\;\;E\subseteq\{1,...,n\}\nonumber\\
&&{\widetilde  {\cal M}}(A_1,...,A_n)=\sum _{E} (-1)^{n-|E|}\;p\left (\bigcap _{i\in E} A_i\right )
\end{eqnarray}
where $|E|$ is the cardinality of $E$. For $n=1$, we get ${\cal M}(A_1)={\widetilde  {\cal M}}(A_1)=p(A_1)$.
The quantities
\begin{eqnarray}\label{20}
&&\delta(A_1,...,A_n)={\cal M}(A_1,...,A_n)+(-1)^n\;p(A_1\cap ...\cap A_n)=0\nonumber\\
&&{\widetilde  \delta}(A_1,...,A_n)={\widetilde {\cal M}}(A_1,...,A_n)+(-1)^n\;p(A_1\cup ...\cup A_n)=0.
\end{eqnarray}
generalize the quantities in Eqs.(\ref{1}), (\ref{112}).
Indeed, Eq.(\ref{20}) reduces to Eqs.(\ref{1}), (\ref{112}), for $E=\{A_1,A_2\}$ and $E=\{A_1,A_2,A_3\}$, correspondingly.
We refer to them as M\"obius quantities or non-additivity quantities (because they are equal to zero in the case of additive Kolmogorov probabilities).

\paragraph *{Distributivity and the law of total probability in Kolmogorov's probability theory:}
If $B_1,...,B_n$ is a partition of the set $\Omega$ then using the distributivity property we get
\begin{eqnarray}
A=A\cap \Omega=A\cap (B_1\cup...\cup B_n)=(A\cap B_1)\cup...\cup (A\cap B_n).
\end{eqnarray}
Since $(A\cap B_i)\cap(A\cap B_j)=\emptyset$, and using the additivity property of Kolmogorov probabilities, we get the law of total probability
\begin{eqnarray}
\Delta (A;B_1,...,B_n)=p(A)-p(A\cap B_1)-...-p(A\cap B_n)=0.
\end{eqnarray}
This law is central in Kolmogorov's probability theory and it is based on both the distributivity property of set theory, and the additivity property 
(Eq.(\ref{1})). In the case of partition of the set $\Omega$, into $B$ and its complement ${\overline B}$,
the above relation becomes
\begin{eqnarray}\label{456}
\Delta (A;B)=p(A)-p(A\cap B)-p(A\cap {\overline B})=0.
\end{eqnarray}

\paragraph*{The Dempster-Shafer theory\cite{D1,D2,D3,D4,D5,D6} :}This theory assigns two probabilities to a subset $A$ of $\Omega$.
The lower probability (or belief) $\ell (A)$ and the upper probability (or plausibility) $u(A)=1-{ \ell } ( \overline A)$.
It formalizes what in everyday language is called `the worst and best case scenario'.
Kolmogorov probability theory is based on Boolean logic, and
an element of $\Omega$, belongs to either $A$ or to $\overline A$ (the `law of the excluded middle').
In Dempster-Shafer theory there are three categories: `belongs to $A$', `belongs to $\overline A$' and `don't know'.
The lower probability $\ell (A)$ describes the `belongs to $A$', while the upper probability $u(A)$ combines the `belongs to $A$' and the `don't know'.

The lower and upper probabilities do not obey the additivity property of Eq.(\ref{1}), but obey the inequalities
\begin{eqnarray}\label{127}
&&\ell (A\cup B)-\ell (A)-\ell (B)+\ell (A\cap B)\ge 0\nonumber\\
&&u (A\cup B)-u (A)-u (B)+u (A\cap B)\le 0
\end{eqnarray}
For $B={\overline A}$ these equations reduce to 
\begin{eqnarray}
\ell (A)+\ell (\overline A)\le 1;\;\;\;\;u(A)+u (\overline A)\ge 1.
\end{eqnarray}
The $1-\ell (A)-\ell (\overline {A})=u(A)+u (\overline A)-1$ corresponds to Dempster\rq{}s `don't know' case.
In contrast for Kolmogorov probabilities $p(A)+p (\overline A)=1$.
Also the law of total probability which is based on the additivity property, is not valid for Dempster-Shafer probabilities.

Dempster introduced these probabilities through a multivalued map from a set $S_1$ to another set $S_2$. He showed that due to multivaluedness, standard (Kolmogorov) probabilities in the set $S_1$, become lower and upper probabilities in $S_2$. In a quantum context, Dempster multivaluedness is the fact that  a classical product of two quantities 
$\theta \phi$ becomes a product of two non-commuting operators which can be ordered as ${\hat \theta} 
{\hat \phi}$ or ${\hat \phi} {\hat \theta}$ or in between.

\subsection{The modular orthocomplemented lattice $\Lambda _d$ of subspaces of $H(d)$}

We consider a quantum system $\Sigma (d)$ with variables in ${\mathbb Z}(d)$ (the integers modulo $d$), with states in a $d$-dimensional
Hilbert space $H(d)$\cite{vour,vour2}.
We also consider the basis of `position states' $|{X};m\rangle$ 
\begin{eqnarray}
\ket{X;0}=(1\;0...0)^T;\;\;\;\;\;\ket{X;d-1}=(0\; 0...1)^T.
\end{eqnarray}
The basis of `momentum states' $|{P};m\rangle$ is defined through a finite Fourier transform:
\begin{equation}\label{PPPT}
|{P};n\rangle=F|{X};n\rangle;\;\;\;\;
F=d^{-1/2}\sum _{m,n}\omega (mn)\ket{X;m}\bra{X;n};\;\;\;\;
\omega(m)=\exp \left (i\frac {2\pi m}{d}\right ).
\end{equation}
The $X,P$ in the notation indicate position and momentum states.

The set of subspaces of $H(d)$ with the logical operations\cite{la1,la2,la3,la4,la5}
\begin{eqnarray}
H_1\wedge H_2=H_1\cap H_2;\;\;\;\;\;H_1\vee H_2={\rm span}(H_1 \cup H_2)
\end{eqnarray}
is the Birkhoff-von Neumann orthomodular lattice \cite{LO1,LO2,LO3,LO4}.
These two operations define the logical `AND' and `OR' in a quantum context.
In our case of finite Hilbert spaces, the lattice is modular orthocomplemented lattice and we call it $\Lambda _d$.
The corresponding partial order $\prec$ is `subspace'.
The smallest element in $\Lambda _d$ is ${\cal O}=H(0)$ (the zero-dimensional subspace that contains only the zero vector), and 
the largest element is ${\cal I}=H(d)$. 

We use the notation $\Pi(H_1)$ for the projector to the subspace $H_1$.
Projectors become probabilities by taking the trace of their product with a density matrix $\rho$.
If the dimension of $H_1$ is $d_1\ge 2$, the $\Pi(H_1)$ is a cumulative projector,
in the sense that in the corresponding probabilities there is a 
variable (related to a basis in $H_1$) which takes a range of values ($d_1$ values). 
This is an extension of the term cumulative probabilities, where the random variable takes a range of values.

$H_1^{\perp}$ denotes the orthocomplement of $H_1$, and by definition it obeys the properties
\begin{eqnarray}
&&H_1\wedge H_1^{\perp}={\cal O};\;\;\;\;H_1\vee H_1^{\perp}={\cal I};\;\;\;\;(H_1^{\perp})^{\perp}=H_1\nonumber\\
&&(H_1\wedge H_2)^{\perp}=H_1^{\perp}\vee H_2^{\perp};\;\;\;\;(H_1\vee H_2)^{\perp}=H_1^{\perp}\wedge H_2^{\perp}.
\end{eqnarray}
We say that the space $H_1$ commutes with $H_2$, and we denote this as $H_1{\cal C} H_2$ if
\begin{eqnarray}\label{10}
H_1=(H_1\wedge H_2)\vee (H_1\wedge H_2^{\perp}).
\end{eqnarray}
In orthomodular lattices (like $\Lambda _d$) if $H_1{\cal C} H_2$ then $H_2{\cal C} H_1$.
Also $H_1^{\perp}{\cal C} H_1$.
Furthermore $H_1{\cal C} H_2$ if and only if $[\Pi(H_1 ),\Pi(H_2 )]=0$.

\section{M\"obius operators}\label{s1}

In analogy to $\delta (A_1,A_2)$,  we have introduced in \cite{VO1} the following non-additivity operator:
\begin{eqnarray}\label{32}
{\mathfrak D}(H_1, H_2)=\Pi(H_1\vee H_2)+\Pi(H_1\wedge H_2)-\Pi(H_1)-\Pi(H_2);\;\;\;\;\;{\rm Tr}[{\mathfrak D}(H_1, H_2)]=0.
\end{eqnarray} 
We have proved in \cite{VO1} that
the commutator $[\Pi(H_1),\Pi(H_2)]$ is related to ${\mathfrak D}(H_1, H_2)$, through the relation:
\begin{eqnarray}\label{e3}
[\Pi (H_1),\Pi(H_2)]={\mathfrak D}(H_1, H_2)[\Pi(H_1)-\Pi(H_2)].
\end{eqnarray}
This relation links directly non-commutativity with non-additive probabilities.
In sublattices of $\Lambda _d$ which are Boolean algebras $[\Pi (H_1),\Pi(H_2)]=0$, and ${\mathfrak D}(H_1, H_2)=0$ and 
probabilities can be interpreted as additive (Kolmogorov) probabilities.  
But in the full lattice, in general the projectors do not commute, the ${\mathfrak D}(H_1, H_2)$ is non-zero, and the corresponding probabilities are non-additive.

We interpret the probabilities ${\rm Tr}[\rho \Pi(H_1)]$, ${\rm Tr}[\rho \Pi(H_2)]$ as upper or lower probabilities
within the Dempster-Shafer theory, according to whether the ${\rm Tr}[\rho{\mathfrak D}(H_1, H_2)]$ is positive or negative, correspondingly (see Eqs.(\ref{127})).
 Below we generalize this to M\"obius (or non-additivity) operators with many subspaces, using M\"obius transform in analogous way to Eq.(\ref{20}):
\begin{eqnarray}\label{106}
{\mathfrak D}(H_1,...,H_n)=\sum  _{E}(-1)^{n-|E|}
\Pi\left (\bigvee _{i\in E} H_i\right )+(-1)^n\Pi\left (\bigwedge _{i\in E} H_i\right );\;\;\;\;
E\subseteq\{1,...,n\}.
\end{eqnarray}
Also
\begin{eqnarray}\label{107}
{\widetilde {\mathfrak D}}(H_1,...,H_n)=\sum  _{E}(-1)^{n-|E|}
\Pi\left (\bigwedge _{i\in E} H_i\right )+(-1)^n\Pi\left (\bigvee _{i\in E} H_i\right ).
\end{eqnarray}

Examples are the operator ${\mathfrak D}(H_1,H_2)$ of Eq.(\ref{32}), the
\begin{eqnarray}\label{df1}
{\mathfrak D} (H_1, H_2, H_3)&=&
\Pi(H_1\vee H_2 \vee H_3)
-\Pi(H_1\vee H_2 )
-\Pi(H_1\vee H_3)\nonumber\\
&-&\Pi(H_2 \vee H_3)
+\Pi(H_1)+\Pi(H_2 )+\Pi(H_3)-\Pi(H_1\wedge H_2 \wedge H_3)
\end{eqnarray}
and
\begin{eqnarray}\label{df2}
{\widetilde  {\mathfrak D}} (H_1, H_2, H_3)&=&
\Pi(H_1\wedge H_2 \wedge H_3)
-\Pi(H_1\wedge H_2 )
-\Pi(H_1\wedge H_3)\nonumber\\
&-&\Pi(H_2 \wedge H_3)
+\Pi(H_1)+\Pi(H_2 )+\Pi(H_3)-\Pi(H_1\vee H_2 \vee H_3).
\end{eqnarray}
These operators are the analogues of $\delta (A_1, A_2, A_3)$ and ${\widetilde \delta}(A_1,A_2,A_3)$.
Both the ${\mathfrak D} (H_1, H_2, H_3)$ and ${\widetilde {\mathfrak D}} (H_1, H_2, H_3)$
are  symmetric, i.e., they do not change with permutations of the $H_1, H_2, H_3$. 

\begin{remark}
\mbox{}
\begin{itemize}
\item
Classical `AND' and `OR' logical operations are defined on Boolean algebras.
The relevant probability theory is Kolmogorov's theory, and it is based on set theory which is a Boolean algebra.
The disjunction is the union of two subsets.
The M\"obius transform in Eq.(\ref{2}) uses probabilities corresponding to unions and intersections between sets.
Using the distributivity property of Boolean algebras, we prove that the M\"obius quantities ${\delta}(A_1,...,A_n)$
and ${\widetilde  \delta}(A_1,...,A_n)$ are equal to zero.
\item
Quantum `AND' and `OR' logical operations are defined on the Birkhoff-von Neumann modular lattice $\Lambda _d$.
The disjunction is much more than the union of two subspaces, because it contains all superpositions of vectors in the two subspaces.
Quantum probabilities are based on projectors.
The M\"obius transform in Eq.(\ref{106}),(\ref{107}) uses projectors corresponding to disjunctions and conjunctions of subspaces.
The lack of the distributivity property (only modularity holds which is a weak version of distributivity) implies that
the M\"obius projectors ${\mathfrak D}(H_1,...,H_n)$ and ${\widetilde {\mathfrak D}}(H_1,...,H_n)$ 
are non-zero. We have seen in Eq.(\ref{e3}) that there is a link between M\"obius projectors and non-commutativity, and 
this is extended further below. We also define non-distributivity projectors which quantify the lack of distributivity.
\end{itemize}
\end{remark}

The following proposition provides relations between the M\"obius operators and links them to commutators of $\Pi (H_1)$, $\Pi(H_2)$, $\Pi(H_3)$.
In particular, Eq.(\ref{bg}) (together with Eq.(\ref{e3})) show that the M\"obius operators are intimately connected to non-commutativity.
\begin{proposition}\label{12}
\mbox{}
\begin{itemize}
\item[(1)]
If $H_1\prec H_2$ then ${\mathfrak D} (H_1, H_2, H_3)=-{\mathfrak D} (H_1, H_3)$
and ${\widetilde  {\mathfrak D}} (H_1, H_2, H_3)=-{\mathfrak D} (H_2, H_3)$.
\item[(2)]
If the $H_1,H_2,H_3$ belong to the same chain (i.e., $H_i\prec H_j\prec H_k$ where $\{i,j,k\}=\{1,2,3\}$) then 
${\mathfrak D} (H_1, H_2, H_3)={\widetilde  {\mathfrak D}} (H_1, H_2, H_3)=0$.
\item[(3)]
\begin{eqnarray}\label{315}
{\mathfrak D} (H_1, H_2, H_3)+{\widetilde  {\mathfrak D}} (H_1, H_2, H_3)
+{\mathfrak D} (H_1, H_2)+{\mathfrak D} (H_1, H_3)+{\mathfrak D} (H_2, H_3)=0
\end{eqnarray}
\item[(4)]
\begin{eqnarray}\label{33}
\Pi (H_1)\Pi (H_3)\Pi (H_2)-\Pi(H_1\wedge H_2 \wedge H_3)=\Pi (H_1){\mathfrak D} (H_1, H_2, H_3)\Pi (H_2).
\end{eqnarray}
\item[(5)]
\begin{eqnarray}\label{bg}
\left [[\Pi (H_1),\Pi(H_3)],\Pi(H_2)\right ]&=&
[\Pi (H_1)-\Pi (H_3)]{\mathfrak D} (H_1, H_2, H_3)\Pi (H_2)\nonumber\\&+&\Pi (H_2){\mathfrak D} (H_1, H_2, H_3)[\Pi (H_1)-\Pi (H_3)].
\end{eqnarray}

\end{itemize}
\end{proposition}
\begin{proof}
\mbox{}
\begin{itemize}
\item[(1)]
If $H_1\prec H_2$ then
\begin{eqnarray}
\Pi(H_1\vee H_2 \vee H_3)=\Pi(H_2 \vee H_3);\;\;\;\;\;
\Pi(H_1\vee H_2 )=\Pi(H_2 );\;\;\;\;\;
\Pi(H_1\wedge H_2 \wedge H_3)=\Pi(H_1\wedge H_3).
\end{eqnarray}
Using this we prove that ${\mathfrak D} (H_1, H_2, H_3)=-{\mathfrak D} (H_1, H_3)$.
In analogous way we prove that ${\widetilde {\mathfrak D}} (H_1, H_2, H_3)=-{\mathfrak D} (H_2, H_3)$.
\item[(2)]
If $H_i\prec H_j\prec H_k$ then according to the first part of this proposition ${\mathfrak D} (H_i, H_j, H_k)=-{\mathfrak D} (H_i, H_k)$.
But we have also proved in  \cite{VO1}, that 
\begin{eqnarray}\label{333}
H_i \prec H_k\;\;\rightarrow \;\;{\mathfrak D} (H_i, H_k)=0.
\end{eqnarray}
In analogous way we prove that ${\widetilde {\mathfrak D}} (H_1, H_2, H_3)=0$.
\item[(3)]
We add Eqs.(\ref{df1}),(\ref{df2}) and get Eq.(\ref{315}).
\item[(4)]
We multiply ${\mathfrak D} (H_1, H_2, H_3)$ with $\Pi (H_1)$ on the left, and $\Pi (H_2)$ on the right, and we get Eq.(\ref{33}).
\item[(5)]
Using Eq.(\ref{33}) we prove that
\begin{eqnarray}
&&\Pi (H_1){\mathfrak D} (H_1, H_2, H_3)\Pi (H_2)+\Pi (H_2){\mathfrak D} (H_1, H_2, H_3)\Pi (H_1)\nonumber\\&&=\Pi (H_1)\Pi (H_3)\Pi (H_2)+\Pi (H_2)\Pi (H_3)\Pi (H_1)-2\Pi(H_1\wedge H_2 \wedge H_3).
\end{eqnarray}
Using this we can now prove Eq.(\ref{bg}). 

\end{itemize}

\end{proof}

\begin{example}
We consider the $3$-dimensional space $H(3)$ and its one-dimensional subspaces $H_1$, $H_2$, $H_3$, defined with the 
following vectors, in the position representation:
\begin{eqnarray}\label{258}
v_1=(0.3, 0.3, 0.905)^T;\;\;\;\;v_2=(0.4,0.5,0.768)^T;\;\;\;\;v_3=\frac{v_1+v_2}{|v_1+v_2|}=(0.353,0.403,0.844)^T.
\end{eqnarray} 
We find the $\Pi (H_1\vee H_2)$ as follows.
The vector 
\begin{eqnarray}
v_2'=[1-\Pi(H_1)]v_2=v_2-v_1(v_1^{\dagger} v_2),
\end{eqnarray}
is perpendicular to $v_1$, and it is on the plane defined by the $v_1,v_2$.
We normalize $v_2'$ into $w_2=v_2'/|v_2'|$ and then $\Pi(H_1\vee H_2)=\Pi(H_1)+w_2w_2^{\dagger}$.
In analogous way we calculate the $\Pi(H_1\vee H_3)$, etc.
Also
\begin{eqnarray}\label{7b}
H_1\wedge H_2 =H_2 \wedge H_3=H_1\wedge H_3={\cal O}.
\end{eqnarray}

In this case
\begin{eqnarray}\label{ex10}
&&{\mathfrak D}(H_1, H_2)=
\left (
\begin{array}{ccc}
0.019&0.142&-0.480\\
0.142&0.403&-0.714\\
-0.480&-0.714&-0.422\\
\end{array}
\right );\;\;\;\;
{\mathfrak D}(H_1, H_3)=
\left (
\begin{array}{ccc}
0.055&0.200&-0.471\\
0.200&0.490&-0.671\\
-0.471&-0.671&-0.545\\
\end{array}
\right )\nonumber\\&&
{\mathfrak D}(H_2, H_3)=
\left (
\begin{array}{ccc}
-0.014&0.090&-0.506\\
0.090&0.330&-0.783\\
-0.506&-0.783&-0.316\\
\end{array}
\right )
\end{eqnarray}
Also
\begin{eqnarray}\label{ex11}
{\mathfrak D}(H_1, H_2, H_3)=
\left (
\begin{array}{ccc}
-0.054&-0.210&0.457\\
-0.210&-0.539&0.668\\
0.457&0.668&0.593\\
\end{array}
\right );\;\;\;\;\;
{\widetilde  {\mathfrak D}}(H_1, H_2, H_3)=
\left (
\begin{array}{ccc}
-0.006&-0.222&1.000\\
-0.222&-0.685&1.499\\
1.000&1.499&0.691\\
\end{array}
\right ).
\end{eqnarray}
If quantum probabilities were Kolmogorov probabilities, all these matrices would have been zero.
In this sense they quantify deviations of quantum probabilities, from the Kolmogorov probability theory.
\end{example}

\section{Non-distributivity and the violation of the law of the total probability}\label{s3}

\subsection{Non-distributivity projectors}
In any lattice the following distributivity inequalities hold:
\begin{eqnarray}\label{29}
(H_1\wedge H_2) \vee H_0 \prec (H_1\vee H_0 )\wedge(H_2\vee H_0)\nonumber\\
(H_1\vee H_2) \wedge H_0\succ (H_1\wedge H_0 )\vee(H_2\wedge H_0)
\end{eqnarray}
In a distributive lattice they become equalities.
Distributivity relations that involve $H_0$ with many subspaces $H_1,...,H_n$, can also be written.

The lattice $\Lambda _d$ is modular but non-distributive. 
The modularity property is
\begin{eqnarray}\label{292}
H_1\prec H_3\;\rightarrow\;H_1\vee(H_2\wedge H_3)=
(H_1\vee H_2) \wedge H_3
\end{eqnarray}
and is weaker than distributivity (every distributive lattice is modular).
Examples of distributive lattices are the normal subgroups of any group, and the subspaces of any finite-dimensional vector space.

We introduce the following projectors that measure deviations from distributivity:
\begin{eqnarray}\label{40}
&&{\varpi}_1 (H_1, H_2| H_0)=\Pi[(H_1\vee H_0 )\wedge(H_2\vee H_0)]-\Pi[(H_1\wedge H_2) \vee H_0]\nonumber\\
&&{\varpi}_2 (H_1, H_2| H_0)=
\Pi[(H_1\vee H_2) \wedge H_0]
-\Pi[(H_1\wedge H_0 )\vee(H_2\wedge H_0)]\nonumber\\
&&[{\varpi}_i (H_1, H_2| H_0)]^2 ={\varpi}_i (H_1, H_2| H_0);\;\;\;\;i=1,2
\end{eqnarray}
Clearly these functions do not change, under permutation of $H_1,H_2$.

The following relations are easily proved:
\begin{eqnarray}\label{349}
{\varpi}_1 (H_1, H_2| H_0)&=&
-{\mathfrak D} (H_1 \wedge H_2, H_0)+{\mathfrak D} (H_1 \wedge H_0, H_2\wedge H_0)
-\Pi(H_1\wedge H_2)\nonumber\\
&+&\Pi(H_1\wedge H_0 )
+\Pi(H_2\wedge H_0)-\Pi(H_0)\nonumber\\&=&
-{\mathfrak D} (H_1, H_2, H_0)-{\mathfrak D} (H_1\wedge H_2, H_0)-{\mathfrak D} (H_1, H_2)+{\mathfrak D} (H_1\vee H_0, H_2\vee H_0).
\end{eqnarray}
Also
\begin{eqnarray}
{\varpi}_2 (H_1, H_2| H_0)&=&
{\mathfrak D} (H_1 \vee H_2, H_0)-{\mathfrak D} (H_1 \vee H_0, H_2\vee H_0)
+\Pi(H_1\vee H_2)
\nonumber\\&-&\Pi(H_1\vee H_0 )
-\Pi(H_2\vee H_0)+\Pi(H_0)
\nonumber\\&=&
{\widetilde{\mathfrak D}} (H_1, H_2, H_0)+{\mathfrak D} (H_1\vee H_2, H_0)+{\mathfrak D} (H_1, H_2)-{\mathfrak D} (H_1\wedge H_0, H_2\wedge H_0).
\end{eqnarray}
These equations show that non-distributivity, non-additivity of probability, and non-commutativity are all linked together.
This statement is also strengthened with the following proposition:
\begin{proposition}
\mbox{}
\begin{itemize}
\item[(1)]
${\varpi}_1 (H_1, H_1^{\perp}| H_0)={\varpi}_2 (H_1, H_1^{\perp}| H_0)=0$ if and only if $[\Pi(H_1 ),\Pi(H_0 )]=0$.
\item[(2)]
If $H_1\prec H_2$ then ${\varpi}_1 (H_1, H_2| H_0)={\varpi}_2 (H_1, H_2| H_0)=0$.
\item[(3)]
If $H_1,H_0$ are comparable (i.e. $H_1\prec H_0$ or $H_0\prec H_1$) then ${\varpi}_1 (H_1, H_2| H_0)={\varpi}_2 (H_1, H_2| H_0)=0$.
\item[(4)]
If any two of the projectors $\Pi(H_1 ),\Pi(H_2 ),\Pi(H_0 )$, (each) commute with the third one, then
${\varpi}_1 (H_1, H_2| H_0)={\varpi}_2 (H_1, H_2| H_0)=0$.
\end{itemize}
\end{proposition}
\begin{proof}
\mbox{}
\begin{itemize}
\item[(1)]
In the special case that $H_2=H_1^{\perp}$ the inequalities in Eq.(\ref{29}) become
\begin{eqnarray}\label{300}
(H_1\wedge H_0 )\vee(H_1^{\perp}\wedge H_0)\prec H_0\prec (H_1\vee H_0 )\wedge(H_1^{\perp}\vee H_0),
\end{eqnarray}
and the quantities in Eq.(\ref{40}) become
\begin{eqnarray}\label{39}
{\varpi}_1 (H_1, H_1^{\perp}| H_0)&=&\Pi[(H_1\vee H_0 )\wedge(H_1^{\perp}\vee H_0)]-\Pi(H_0)\nonumber\\
{\varpi}_2 (H_1, H_1^{\perp}| H_0)&=&\Pi(H_0)-\Pi[(H_1\wedge H_0 )\vee(H_1^{\perp}\wedge H_0)].
\end{eqnarray}
If  $[\Pi(H_1 ),\Pi(H_0 )]=0$ or equivalently $H_1 {\cal C} H_0$, the inequalities in Eq.(\ref{300}) become equalities (see Eq.(\ref{10})), and
${\varpi}_1 (H_1, H_1^{\perp}| H_0)={\varpi}_2 (H_1, H_1^{\perp}| H_0)=0$.
The converse is also easily seen to be true.
\item[(2)]
If $H_1\prec H_2$ then ${\mathfrak D} (H_1, H_2, H_0)=-{\mathfrak D} (H_1, H_0)$ (proposition \ref{12}),
and Eq.(\ref{349}) gives
\begin{eqnarray}
\varpi _1 (H_1, H_2| H_0)&=&
{\mathfrak D} (H_1, H_0)-{\mathfrak D} (H_1, H_0)-{\mathfrak D} (H_1, H_2)+{\mathfrak D} (H_1\vee H_0, H_2\vee H_0)\nonumber\\&=&-{\mathfrak D} (H_1, H_2)+{\mathfrak D} (H_1\vee H_0, H_2\vee H_0).
\end{eqnarray}
The fact that $H_1\prec H_2$ implies that $H_1\vee H_0\prec H_2\vee H_0$ and taking into account Eq.(\ref{333}) we see that ${\mathfrak D} (H_1, H_2)={\mathfrak D} (H_1\vee H_0, H_2\vee H_0)=0$.
\item[(3)]
The proof of this part is analogous to the first part.
\item[(4)]
The Foulis-Holland theorem\cite{la1} states that in orthomodular lattices, if $H_1{\cal C} H_0$ and $H_2{\cal C} H_0$ then the sublattice generated by the $H_1, H_2, H_0$ is distributive.
This proves that ${\varpi}_1 (H_1, H_2| H_0)={\varpi}_2 (H_1, H_2| H_0)=0$.

\end{itemize}

\end{proof}
\begin{remark}
If any two of the projectors $\Pi(H_1 ),\Pi(H_2 ),\Pi(H_0 )$ commute with the third one, then 
\begin{eqnarray}
[[\Pi(H_i ),\Pi(H_j )],\Pi(H_k )]=0;\;\;\;\;\{i,j,k\}=\{1,2,3\}.
\end{eqnarray}
We have shown that in this case ${\varpi}_1 (H_1, H_2| H_0)={\varpi}_2 (H_1, H_2| H_0)=0$, and this exemplifies again 
that non-commutativity is linked to non-distributivity.
\end{remark}

\subsection{Deviations from the law of the total probability}

In analogy to Eq.(\ref{456}) we introduce the projectors
\begin{eqnarray}\label{123}
\pi(H_0;H_1)=\Pi(H_0)-\Pi(H_1\wedge H_0 )-\Pi (H_1^{\perp}\wedge H_0)
\end{eqnarray}
which measure deviations from the law of the total probability in Kolmogorov's theory.
\begin{proposition}
\begin{eqnarray}
\pi(H_0;H_1)={\varpi}_2 (H_1, H_1^{\perp}| H_0)+{\mathfrak D}(H_1\wedge H_0, H_1^{\perp}\wedge H_0).
\end{eqnarray}
\end{proposition}
\begin{proof}
The proof is based on Eq.(\ref{39}) in conjunction with the relation
\begin{eqnarray}
{\mathfrak D}(H_1\wedge H_0, H_1^{\perp}\wedge H_0)=\Pi[(H_1\wedge H_0)\vee (H_1^{\perp}\wedge H_0)]-\Pi (H_1\wedge H_0)-\Pi (H_1^{\perp}\wedge H_0).
\end{eqnarray}
\end{proof}
We have explained earlier that the law of the total probability for Kolmogorov probabilities, relies on the distributivity property of set theory and on the additivity of probability.
The above proposition shows that the projector $\pi(H_0;H_1)$ which measures deviations from the law of the total probability,
is equal to the projector ${\varpi}_2 (H_1, H_1^{\perp}| H_0)$ which measures deviations from distributivity, plus the operator
${\mathfrak D}(H_1\wedge H_0, H_1^{\perp}\wedge H_0)$ which measures deviations from additivity.

\subsection{Example}
In $H(3)$ we consider the one-dimensional subspaces $H_1$, $H_2$, $H_3$, defined with the vectors $v_1$, $v_2$, $v_3$,
correspondingly, in Eq.(\ref{258}).
The vector $v_3$ is in the plane that the $v_1, v_2$ define, and therefore $H_1\vee H_3=H_2\vee H_3$.
Taking into account Eq.(\ref{7b}), we find that Eqs.(\ref{40}),(\ref{123}), reduce to
\begin{eqnarray}
{\varpi}_1 (H_1, H_2| H_3)&=&\Pi(H_1\vee H_3)-\Pi(H_3)\nonumber\\
{\varpi}_2 (H_1, H_2| H_3)&=&\Pi(H_3)\nonumber\\
\pi(H_3;H_1)&=&\Pi(H_3).
\end{eqnarray}
Therefore
\begin{eqnarray}\label{59}
&&{\varpi}_1 (H_1, H_2| H_3)=
\left (
\begin{array}{ccc}
0.145&0.290&-0.199\\
0.290&0.580&-0.399\\
-0.199&-0.399&0.275\\
\end{array}
\right )\nonumber\\
&&{\varpi}_2 (H_1, H_2| H_3)=\pi(H_3;H_1)=
\left (
\begin{array}{ccc}
0.125&0.142&0.298\\
0.142&0.163&0.340\\
0.298&0.340&0.712\\
\end{array}
\right ).
\end{eqnarray}
If quantum probabilities were Kolmogorov probabilities, all these matrices would have been zero.

\section{Observables}\label{s4}

The non-additivity operators ${\mathfrak D} (H_1, H_2)$, ${\mathfrak D} (H_1, H_2, H_3)$, etc, are Hermitian operators and therefore they are observables.
The projectors ${\varpi}_1 (H_1, H_2| H_3)$, ${\varpi}_2 (H_1, H_2| H_3)$, $\pi(H_0;H_1)$, are also observables.
In this section we discuss briefly measurements with them, on a system with a density matrix $\rho$.

Measurements with a Hermitian operator $\Theta$ on a system described with a density matrix $\rho$, will give values with average and standard deviation
\begin{eqnarray}\label{1110}
E(\Theta)={\rm Tr}(\rho \Theta);\;\;\;\;\;\Delta (\Theta)=\left \{{\rm Tr}(\Theta ^2 \rho)-[{\rm Tr}(\Theta \rho )]^2\right \}^{1/2}
\end{eqnarray}
We consider measurements with the operators ${\mathfrak D} (H_1, H_2)$, $\Pi(H_1\vee H_2)$, $\Pi(H_1)$, $\Pi(H_2)$, $\Pi(H_1\wedge H_2)$.
The operators ${\mathfrak D} (H_1, H_2)$, $\Pi(H_1)$, $\Pi(H_2)$, do not commute and they are not simultaneously measurable.
They will be measured using different ensembles of the system described by $\rho$.
The operators ${\mathfrak D} (H_1, H_2)$, $\Pi(H_1\vee H_2)$, $\Pi(H_1\wedge H_2)$, do commute and they are simultaneously measurable.

From Eq.(\ref{32}) it follows that the averages are related through the relation
\begin{eqnarray}\label{1110}
E[{\mathfrak D} (H_1, H_2)]=E[\Pi(H_1\vee H_2)]-E[\Pi(H_1)]-E[\Pi(H_2)]+E[\Pi(H_1\wedge H_2)],
\end{eqnarray}
and the standard deviations through the relation
\begin{eqnarray}\label{110}
\{\Delta [{\mathfrak D} (H_1, H_2)]\}^2&=&\{\Delta [\Pi(H_1\vee H_2)]\}^2-\{\Delta [\Pi(H_1)]\}^2-\{\Delta [\Pi(H_2)]\}^2+\{\Delta [\Pi(H_1\wedge H_2)]\}^2+a,
\end{eqnarray}
where
\begin{eqnarray}
a&=&-2\{E[\Pi(H_1)]\}^2-2\{E[\Pi(H_2)]\}^2-2E[\Pi(H_1)]E[\Pi(H_2)]\nonumber\\&+&2E[\Pi(H_1\vee H_2)]\{E[\Pi(H_1)]+E[\Pi(H_2)]\}+E[\Pi(H_1)\Pi(H_2)+\Pi(H_2)\Pi(H_1)]
\nonumber\\&+&2E[\Pi(H_1\wedge H_2)]\{E[\Pi(H_1)]+E[\Pi(H_2)]-E[\Pi(H_1\vee H_2)]-1\}.
\end{eqnarray}
The $E[{\mathfrak D} (H_1, H_2)]$ and $\Delta [{\mathfrak D} (H_1, H_2)]$ can be viewed as quantities that quantify the violation of the additivity
by quantum probabilities related to $\Pi(H_1), \Pi(H_2)$, in a system with density matrix $\rho$.
In the case that the $\Pi(H_1), \Pi(H_2)$ commute, the ${\mathfrak D} (H_1, H_2)=0$, additivity holds, and $E[{\mathfrak D} (H_1, H_2)]=\Delta [{\mathfrak D} (H_1, H_2)]=0$.
Conversely, if $E[{\mathfrak D} (H_1, H_2)]=\Delta [{\mathfrak D} (H_1, H_2)]=0$, and in addition to this all the higher moments are also zero,
then ${\mathfrak D} (H_1, H_2)=0$, additivity holds, and the $\Pi(H_1), \Pi(H_2)$ commute.

Analogous comments can be made for the projectors ${\varpi}_1 (H_1, H_2| H_0)$, ${\varpi}_2 (H_1, H_2| H_0)$, $\pi(H_0;H_1)$.
For $\Theta= {\varpi}_i (H_1, H_2| H_3)$ the $E[{\varpi}_i (H_1, H_2| H_3)]$, $\Delta [{\varpi}_i (H_1, H_2| H_3)]$,
quantify the violation of distributivity by $H_1,H_2,H_3$, in a system with density matrix $\rho$.
If they (and the higher moments) are zero, ${\varpi}_i (H_1, H_2| H_3)=0$ and distributivity holds for $H_1,H_2,H_3$.
For $\Theta= \pi(H_0;H_1)$ the $E[\pi(H_0;H_1)]$, $\Delta [\pi(H_0;H_1)]$,
quantify the deviations from the law of the total probability by $H_0,H_1$, in a system with density matrix $\rho$.

We note that
${\mathfrak D} (H_1^{\perp}, H_2^{\perp})=-{\mathfrak D} (H_1, H_2)$, and therefore
\begin{eqnarray}
E[ {\mathfrak D} (H_1^{\perp}, H_2^{\perp})]=-E[ {\mathfrak D} (H_1, H_2)];\;\;\;\;\;
\Delta [ {\mathfrak D} (H_1^{\perp}, H_2^{\perp})]=\Delta [ {\mathfrak D} (H_1, H_2)].
\end{eqnarray}

\begin{example}
We consider the density matrix
\begin{eqnarray}
\rho =\frac{1}{3}
\left (
\begin{array}{ccc}
1&1&1\\
1&1&1\\
1&1&1\\
\end{array}
\right ).
\end{eqnarray}
We also consider the subspaces $H_1$, $H_2$, $H_3$, defined with the vectors $v_1$, $v_2$, $v_3$ in Eq.(\ref{258}), and taking into account 
our previous calculations in Eqs.(\ref{ex10}),(\ref{ex11}),(\ref{59}) we get
\begin{eqnarray}
&&E[ {\mathfrak D} (H_1, H_2)]=-0.701;\;\;\;\;\Delta [ {\mathfrak D} (H_1, H_2) ]=0.651\nonumber\\
&&E[ {\mathfrak D} (H_1, H_2, H_3)]=0.610;\;\;\;\;\Delta [ {\mathfrak D} (H_1, H_2, H_3) ]=0.792\nonumber\\
&&E[ {\varpi}_1 (H_1, H_2| H_3)]=0.127;\;\;\;\;\Delta [ {\varpi}_1 (H_1, H_2| H_3) ]=0.334\nonumber\\
&&E[ {\varpi}_2 (H_1, H_2| H_3)]=E[\pi (H_3;H_1) ]=0.854\nonumber\\
&&\Delta [ {\varpi}_2 (H_1, H_2| H_3)]=\Delta [ \pi(H_3;H_1) ]=0.353
\end{eqnarray}
If quantum probabilities were Kolmogorov probabilities, all these quantities (and the higher moments) would have been zero.
\end{example}

\section{Constraints on the M\"obius operators by the modularity of the lattice}\label{s5}

Probability theory needs the logical `OR' and the logical `AND' for its axioms, and in this sense it is tacitly related to a lattice.
Kolmogorov probabilities are related to set theory which is a Boolean algebra. 
We have seen that  quantum probabilities do not obey some of the properties of Kolmogorov probabilities, because the lattice $\Lambda _d$ is not a Boolean algebra.
$\Lambda _d$ is a modular lattice and this implies certain constraints on the projectors corresponding to quantum probabilities, which we study in this section. 

If $H_1\prec H_2$, we use the notation $[H_1,H_2]$, for the interval sublattice of $\Lambda _d$ that contains all the spaces $H_0$ such that 
$H_1\prec H_0\prec H_2$. We call ${I}_d$ the set of all such interval sublattices.
\begin{definition}
The interval sublattice $[H_1,H_2]$ is lower transpose of $[H_3,H_4]$ 
or equivalently $[H_3,H_4]$ is upper transpose of $[H_1,H_2]$ if $H_4=H_2\vee H_3$ and $H_1=H_2\wedge H_3$.
We denote this as $[H_1,H_2]\prec _{\rm tr} [H_3,H_4]$ or equivalently,  as $[H_3,H_4]\succ _{\rm tr}[H_1,H_2]$ (the index `tr' indicates the transpose interval partial order, which is different from the subspace partial order).
\end{definition}
We can easily show that
\begin{eqnarray}
&&[H_1,H_2]\prec _{\rm tr}[H_1,H_2]\nonumber\\
&&[H_1,H_2]\prec _{\rm tr}[H_3,H_4]\;\;{\rm and}\;\;[H_3,H_4]\prec _{\rm tr}[H_1,H_2]\;\;\rightarrow\;\;H_1=H_3\;\;{\rm and}\;\;H_2=H_4\nonumber\\
&&[H_1,H_2]\prec _{\rm tr}[H_3,H_4]\;\;{\rm and}\;\;[H_3,H_4]\prec _{\rm tr}[H_5,H_6]\;\;\rightarrow\;\;[H_1,H_2]\prec _{\rm tr}[H_5,H_6].
\end{eqnarray}
Therefore $\prec _{\rm tr}$ is a partial order in the set ${I}_d$. 

If $[H_1,H_2]\prec _{\rm tr}[H_3,H_4]$ (in which case $H_4=H_2\vee H_3$ and $H_1=H_2\wedge H_3$), the
set that contains all the intervals $[H_a,H_a']$ such that $[H_1, H_2]\prec _{\rm tr}[H_a,H_a']\prec _{\rm tr}[H_3,H_4]$,
is an interval in $I_d$, and we denote it as
\begin{eqnarray}
{\mathfrak I}(H_2, H_3)=[[H_1, H_2], [H_3,H_4]]_{\rm tr}
\end{eqnarray}
The inner brackets indicate intervals with respect to the subspace partial order $\prec$, while the
outer brackets indicate intervals with respect to the transpose interval partial order. From the definition, it follows that if 
$[H_1,H_2]\prec _{\rm tr}[H_3,H_4]$ then $[H_1,H_3]\prec _{\rm tr}[H_2,H_4]$, and we can also define the 
\begin{eqnarray}
{\mathfrak I}(H_3, H_2)=[[H_1, H_3], [H_2,H_4]]_{\rm tr};\;\;\;\;H_1=H_2\wedge H_3;\;\;\;\;H_4=H_2\vee H_3
\end{eqnarray}
The ${\mathfrak I}(H_2, H_3)$ is different from the ${\mathfrak I}(H_3, H_2)$.

Transpose intervals is an important concept in modular lattices, and we present some results which are then used to impose constraints on the M\"obius operators.

\subsection{The partial order of transpose intervals in modular lattices}
\begin{lemma}
\mbox{}
\begin{itemize}
\item[(1)]
If $[H_1,H_1']\prec _{\rm tr}[H_a,H_a']\prec _{\rm tr} [H_2,H_2']$ then $H_a=(H_a\vee H_1')\wedge H_2$.
\item[(2)]
The partial order $\prec _{\rm tr}$ is not locally finite.
\end{itemize}
\end{lemma}
\begin{proof}
\mbox{}
\begin{itemize}
\item[(1)]
We assume that $[H_1,H_1']\prec _{\rm tr}[H_a,H_a']\prec _{\rm tr} [H_2,H_2']$. The definition of transpose intervals implies that
\begin{eqnarray}\label{62}
H_a'=H_a\vee H_1';\;\;\;\;\;H_1=H_a\wedge H_1';\;\;\;\;\;H_2'=H_a'\vee H_2;\;\;\;\;\;H_a=H_a'\wedge H_2
\end{eqnarray}
Combining the first and fourth of these relations we get $H_a=(H_a\vee H_1')\wedge H_2$.
\item[(2)]
A partial order is locally finite if every interval $[a,b]$ contains a finite number of elements.
An interval ${\mathfrak I}(H_1,H_2)$
in $I_d$ has an infinite number of elements in general. For example, if $H_1$ is a two-dimensional space and 
$H_1\wedge H_2$ is a zero dimensional space,
then the interval $[H_1\wedge H_2, H_1]$ contains an infinite number of one-dimensional spaces $h$, and all the $[h, h\vee H_2]$
belong to the interval ${\mathfrak I}(H_2,H_1)$.
Therefore $\prec _{\rm tr}$ is not a locally finite partial order.
\end{itemize}
\end{proof}

\begin{remark}
There is much work on locally finite partial orders, after the work by Rota \cite{R1,R2}.
This work is not applicable here, because $\prec _{\rm tr}$ is not a locally finite partial order.
\end{remark}

The following proposition is well known for modular lattices, and we give it in the context of the lattice $\Lambda _d$, without proof.
\begin{proposition}\label{AA}
There is a bijective map between the transpose interval sublattices $[H_1\wedge H_2,H_1]$ and $[H_2,H_1\vee H_2]$, which maps
$h\in [H_1\wedge H_2,H_1]$ into $h'\in [H_2,H_1\vee H_2]$, where
\begin{eqnarray}\label{br}
h'=h\vee H_2;\;\;\;\;\;h=h'\wedge H_1.
\end{eqnarray}
There is also a bijective map between the transpose interval sublattices $[H_1\wedge H_2,H_2]$ and $[H_1,H_1\vee H_2]$, which maps
${\mathfrak h}\in [H_1\wedge H_2,H_2]$ into ${\mathfrak h}'\in [H_1,H_1\vee H_2]$, where
\begin{eqnarray}\label{br}
{\mathfrak h}'={\mathfrak h}\vee H_1;\;\;\;\;\;{\mathfrak h}={\mathfrak h}'\wedge H_2.
\end{eqnarray}
\end{proposition}

\begin{proposition}\label{pr}
If $h\in [H_1\wedge H_2,H_1]$ then $[h,h\vee H_2]$ belongs to the interval ${\mathfrak I}(H_2, H_1)$ in $I_d$.
Also if ${\mathfrak h}\in [H_1\wedge H_2,H_2]$  then $[{\mathfrak h},{\mathfrak h}\vee H_1]$ 
belongs to the interval ${\mathfrak I}(H_1, H_2)$ in $I_d$.
\end{proposition}
\begin{proof}
We need to prove that if $h\in [H_1\wedge H_2,H_1]$ then $[H_1\wedge H_2, H_2] \prec _{\rm tr}[h,h\vee H_2]\prec _{\rm tr}[H_1, H_1\vee H_2]$.
It is sufficient to prove that
\begin{eqnarray}\label{35}
H_2\wedge h=H_1 \wedge H_2;\;\;\;\;H_1\wedge (h\vee H_2)=h;\;\;\;\;\;H_1\vee (h\vee H_2)=H_1\vee H_2
\end{eqnarray}
In order to prove the first of these equations we note that
\begin{eqnarray}
&&h\prec H_1\;\;\rightarrow\;\;h\wedge H_2\prec H_1\wedge H_2\nonumber\\
&&h\succ H_1\wedge H_2\;\;\rightarrow\;\;h\wedge H_2\succ H_1\wedge H_2
\end{eqnarray}
From this follows that $h\wedge H_2= H_1\wedge H_2$.

In order to prove the second of Eqs(\ref{35}), we use the modularity property, according to which the fact that 
$h\prec H_1$ implies that $H_1\wedge (h\vee H_2)=h\vee (H_1\wedge H_2)$.
Furthermore since $H_1\wedge H_2 \prec h$ we get $H_1\wedge (h\vee H_2)=h$.

The third of Eqs(\ref{35}) follows immediately from the fact that $h\prec H_1$.
\end{proof}
\begin{definition}
If $[H_1,H_2]$ is either lower transpose or upper transpose of $[H_3,H_4]$ we say that $[H_1,H_2]$ is transpose of $[H_3,H_4]$.
Two interval sublattices $[H_A,H_B]$ and $[H_A',H_B']$ are called projective, if there is a finite sequence of intervals
$[h_{Ai},h_{Bi}]$ with $i=1,...,n$ and $h_{A1}=H_A$, $h_{B1}=H_B$, $h_{An}=H_A'$, $h_{Bn}=H_B'$, such that any pair of two successive intervals 
$[h_{Ai},h_{Bi}]$ and $[h_{A(i+1)},h_{B(i+1)}]$, is transpose. 
\end{definition}
The property `transpose' is stronger than `projective'.
If $[H_1,H_1']\prec _{\rm tr} [H_2,H_2']$ and $[H_2,H_2']\succ _{\rm tr} [H_3,H_3']$, 
the $[H_1,H_1']$ and $[H_3,H_3']$ are projective.
In this case
\begin{eqnarray}\label{345}
H_2'=H_1'\vee H_2=H_3'\vee H_2;\;\;\;\;\;H_1=H_1'\wedge H_2;\;\;\;\;\;H_3=H_2\wedge H_3'.
\end{eqnarray}
There is a bijective map between the projective intervals $[H_1, H_1']$ and $[H_3,H_3']$, which maps
$h\in [H_1,H_1']$ into $h'\in [H_3,H_3']$, where
\begin{eqnarray}\label{67}
h'=(h\vee H_2)\wedge H_3';\;\;\;\;\;h=(h'\vee H_2)\wedge H_1'.
\end{eqnarray}

\subsection{Implications for the M\"obius operators}
\begin{proposition}\label{P1}
The ${\mathfrak D}(H_1, H_2)$ has $d$ real eigenvalues, whose sum is equal to zero. At least $d-{\rm dim}(H_1\vee H_2)$ of these eigenvalues are equal to zero.
\end{proposition}
\begin{proof}
${\mathfrak D}(H_1, H_2)$ is a $d\times d$ Hermitian matrix, and we have proved \cite{VO1} that its trace is equal to zero.
Therefore the ${\mathfrak D}(H_1, H_2)$ has $d$ real eigenvalues, whose sum is equal to zero.
The orthocomplement $(H_1\vee H_2)^{\perp}$ is a space orthogonal to $H_1\vee H_2$, with dimension $d-{\rm dim}(H_1\vee H_2)$.
Every vector $v$ in this space is orthogonal to all vectors in $H_1\vee H_2$, and therefore
\begin{eqnarray}
\Pi (H_1)v=\Pi(H_2)v=\Pi(H_1\vee H_2)v=\Pi(H_1\wedge H_2)v=0\;\;\rightarrow\;\;{\mathfrak D}(H_1, H_2)v=0.
\end{eqnarray}
All these vectors are eigenvectors of ${\mathfrak D}(H_1, H_2)$ with corresponding eigenvalue $0$.
The multiplicity of the eigenvalue $0$ is at least $d-{\rm dim}(H_1\vee H_2)$because the dimension of  
$(H_1\vee H_2)^{\perp}$ is $d-{\rm dim}(H_1\vee H_2)$.
\end{proof}
We consider the map 
\begin{eqnarray}
{\mathfrak P}:\;[H_1,H_2]\;\;\rightarrow \;\;{\mathfrak P}([H_1,H_2])=\Pi(H_2)-\Pi(H_1);\;\;\;\;\;[{\mathfrak P}([H_1,H_2])]^2={\mathfrak P}([H_1,H_2])
\end{eqnarray}
which attaches the projector ${\mathfrak P}([H_1,H_2])=\Pi(H_2)-\Pi(H_1)$ to the interval $[H_1,H_2]$.
We also consider the map 
\begin{eqnarray}
\Psi:\;\;{\mathfrak I}(H_1, H_2)\;\;\rightarrow\;\;\Psi[{\mathfrak I}(H_1, H_2)]={\mathfrak P}([H_2,H_1\vee H_2])-
{\mathfrak P}([H_1\wedge H_2, H_1])={\mathfrak D}(H_1, H_2)
\end{eqnarray}
which attaches the non-additivity operator ${\mathfrak D}(H_1, H_2)$ to the interval ${\mathfrak I}(H_1, H_2)$.
It is easily seen that $\Psi[{\mathfrak I}(H_2, H_1)]=\Psi[{\mathfrak I}(H_1, H_2)]$.

\begin{proposition}\label{P2}
\mbox{}
\begin{itemize}
\item[(1)]
If $[H_1,H_1']\prec _{\rm tr}[H_a,H_a']\prec _{\rm tr} [H_3,H_3']$  and $[H_1,H_1']\prec _{\rm tr}[H_b,H_b']\prec _{\rm tr} [H_2,H_2']$,
then
\begin{eqnarray}\label{5v}
{\mathfrak D}(H_1',H_a)+{\mathfrak D}(H_a',H_2)={\mathfrak D}(H_1',H_b)+{\mathfrak D}(H_b',H_2)={\mathfrak D}(H_1',H_2).
\end{eqnarray}
\item[(2)]
If $h\in [H_1\wedge H_2,H_1]$ then
\begin{eqnarray}
{\mathfrak D}(H_2,h)+{\mathfrak D}(h\vee H_2,H_1)={\mathfrak D}(H_2,H_1).
\end{eqnarray}
\end{itemize}
\end{proposition}
\begin{proof}
\mbox{}
\begin{itemize}
\item[(1)]
We get
\begin{eqnarray}
{\mathfrak D}(H_1',H_a)+{\mathfrak D}(H_a',H_2)&=&\Pi(H_a')+\Pi(H_1)-\Pi(H_a)-\Pi(H_1')\nonumber\\&+&\Pi(H_2')+\Pi(H_a)-\Pi(H_a')-\Pi(H_2)
={\mathfrak D}(H_1',H_2).
\end{eqnarray}
Similar proof holds with the interval $[H_b,H_b']$.
\item[(2)]
We combine the first part of the proposition, with proposition \ref{pr}, and we prove the second part.
\end{itemize}
\end{proof}
\begin{proposition}\label{P3}
If $[H_1,H_1']\prec _{\rm tr}[H_2,H_2']$ and $[H_2,H_2']\succ _{\rm tr} [H_3,H_3']$  
then the following relation holds for the projective intervals $[H_1,H_1']$ and $[H_3,H_3']$:
\begin{eqnarray}\label{5v}
{\mathfrak P}([H_3,H_3'])-{\mathfrak P}([H_1,H_1'])=
\Pi(H_3')-\Pi(H_3)-\Pi(H_1')+\Pi(H_1)={\mathfrak D}(H_1',H_2)-{\mathfrak D}(H_2,H_3').
\end{eqnarray}
More generally, if $h\in [H_1,H_1']$ then
\begin{eqnarray}\label{5v}
{\mathfrak P}([H_3,(h\vee H_2)\wedge H_3'])-{\mathfrak P}([H_1,h])&=&
\Pi[(h\vee H_2)\wedge H_3']-\Pi(H_3)-\Pi(h)+\Pi(H_1)\nonumber\\&=&{\mathfrak D}(h,H_2)-{\mathfrak D}[H_2,(h\vee H_2)\wedge H_3'].
\end{eqnarray}
\end{proposition}
\begin{proof}
The proof of the first part is based on Eq.(\ref{345}).
For the second part we take into account Eq.(\ref{67}).
\end{proof}

\section{Coherent projectors}\label{s10}

In this section we use some of the projectors that we introduced in the previous sections, in the context of coherent states.
Using one-dimensional subspaces corresponding to coherent states, we construct projectors to larger spaces, with the following coherence properties:
\begin{itemize}
\item
resolution of the identity
\item
under displacement transformations they transform to operators of the same type
\end{itemize} 
This generalizes coherence to multi-dimensional structures.

\subsection{Coherent states in finite quantum systems}

In the Hilbert space $H(d)$ we consider the displacement operators
\begin{eqnarray}\label{cb}
&&{D}(\alpha, \beta)={Z}^\alpha \;{X}^\beta \;\omega (-2^{-1}\alpha \beta)\nonumber\\
&&{Z}^ \alpha =\sum _{n}\omega (n \alpha)|{X};n\rangle \langle {X};n|;\;\;\;\;\;
{X}^\beta=\sum _{n}\omega (-n \beta)|{P};n\rangle \langle {P};n|\nonumber\\
&&X^\beta \;Z^\alpha = Z^\alpha \;X^\beta \;\omega (-\alpha \beta);\;\;\;\;\;X^d=Z^d={\bf 1},
\end{eqnarray}
In this section we assume that $d$ is an odd integer so that the $2^{-1}$, which is used in Eq.(\ref{cb}),
is defined within ${\mathbb Z}(d)$.

Acting with $D(\alpha , \beta)$ on a normalized reference vector (which is sometimes called `fiducial vector') 
\begin{eqnarray}
\ket {{\mathfrak f}}=\sum _m{\mathfrak f}_m\ket{X;m};\;\;\;\;\sum _m|{\mathfrak f}_m|^2=1,
\end{eqnarray}
we get the $d^2$ coherent states\cite{COH,COH1}:
\begin{eqnarray}\label{coh}
\ket{C;\alpha, \beta}=D(\alpha , \beta)\ket{\mathfrak f};\;\;\;\;\alpha , \beta \in {\mathbb Z}(d).
\end{eqnarray}
The $C$ in the notation indicates coherent states.
The overlap of two coherent states is in general non-zero:
\begin{eqnarray}\label{la}
&&\lambda (\alpha, \beta ;\gamma, \delta)=\bra{C;\alpha, \beta}C;\gamma, \delta\rangle=
\omega [2^{-1}(\alpha \beta+\gamma \delta)-\alpha \delta]\sum _{n}{\mathfrak f}_{n+\delta -\beta}^*{\mathfrak f}_n\omega [n(\gamma -\alpha )]\nonumber\\
&&\langle X;n \ket {C;\alpha, \beta}=\omega (-2^{-1}\alpha \beta+\alpha n){\mathfrak f}_{n-\beta}.
\end{eqnarray}
We consider the one-dimensional spaces $H(\alpha,\beta)$ that contain the coherent states $\ket{C;\alpha,\beta}$,
and we use the simpler notation $\Pi(\alpha,\beta)$ for $\Pi[H(\alpha,\beta)]$:
\begin{eqnarray}\label{1111}
&&\Pi ({\alpha,\beta})=\ket{C;\alpha,\beta}\bra{C;\alpha,\beta};\;\;\;\;\frac{1}{d}\sum _{\alpha,\beta}\Pi({\alpha,\beta})={\bf 1}\nonumber\\
&&D(\kappa, \lambda)\Pi({\alpha,\beta})[D(\kappa, \lambda)]^{\dagger}=\Pi({\alpha +\kappa,\beta +\lambda})
\end{eqnarray}
Also
\begin{eqnarray}
\Pi ^{\perp}({\alpha,\beta})={\bf 1}-\Pi ({\alpha,\beta});\;\;\;\;\;\frac{1}{d(d-1)}\sum _{\alpha,\beta}\Pi ^{\perp}({\alpha,\beta})={\bf 1}.
\end{eqnarray}
They are projectors to $(d-1)$-dimensional spaces.
The $\Pi ({\alpha,\beta})$ overlap with each other and do not commute with each other:
\begin{eqnarray}
{\rm Tr}[\Pi ({\alpha,\beta})\Pi ({\gamma,\delta})]=|\sum _{n}{\mathfrak f}_{n+\delta -\beta}^*{\mathfrak f}_n\omega [n(\gamma -\alpha )]|^2;\;\;\;\;\;
[\Pi ({\alpha,\beta}),\Pi ({\gamma,\delta})]\ne 0.
\end{eqnarray}
This is related to the overcompleteness of the coherent states ($d^2$ vectors 
in a $d$-dimensional space).

Below we use the simpler notation $\Pi({\alpha _1,\beta _1};...;{\alpha _{i},\beta _{i}})$ for 
$\Pi[H({\alpha _1,\beta _1})\vee...\vee H({\alpha _{i},\beta _{i}})]$, and
${\mathfrak D} ({\alpha _1,\beta _1};...;{\alpha _{i},\beta _{i}})$ for
${\mathfrak D} [H({\alpha _1,\beta _1});...;H({\alpha _{i},\beta _{i}})]$.

\subsection{Cumulative coherent projectors for aggregations of coherent states}

We consider the projectors $\Pi ({\alpha _1,\beta _1};{\alpha _2,\beta _2})$ to the two-dimensional spaces $H({\alpha _1,\beta _1}) \vee H({\alpha _2,\beta _2})$.
We call the $\Pi ({\alpha _1,\beta _1};{\alpha _2,\beta _2})$ cumulative projectors because 
they project into two-dimensional spaces, which means that the probabilities corresponding to a basis in this space,
take many values (in this case two values). 
We note that $H({\alpha _1,\beta _1}) \wedge H({\alpha _2,\beta _2})={\cal O}$.

We prove that
\begin{eqnarray}\label{35}
\Pi({\alpha _1,\beta _1};{\alpha _2,\beta _2})&=&{\tau _2}
\left[\Pi({\alpha _1,\beta _1})+\Pi({\alpha _2,\beta _2})
-\Pi({\alpha _1,\beta _1})\Pi({\alpha _2,\beta _2})-
\Pi({\alpha _2,\beta _2})\Pi({\alpha _1,\beta _1})\right ]\nonumber\\
\tau _2&=&\{1-{\rm Tr}[\Pi ({\alpha _1,\beta _1})\Pi ({\alpha _2,\beta _2})]\}^{-1}=\{1-|\lambda ({\alpha _1,\beta _1} ;{\alpha _2,\beta _2})|^2\}^{-1}
\end{eqnarray}
We prove this using the Gram-Schmidt orthogonalization method. We take the component of $\ket {C;{\alpha _2,\beta _2}}$ which is perpendicular to $\ket {C;{\alpha _1,\beta _1}}$,
normalize it into the vector $\ket{s}$ with length equal to one, and then add $\Pi({\alpha _1,\beta _1})$ and $\ket{s}\bra{s}$: 
\begin{eqnarray}
&& \Pi ({\alpha _1,\beta _1};{\alpha _2,\beta _2})=\Pi({\alpha _1,\beta _1})+\ket{u_2}\bra{u_2}\nonumber\\
&& \ket{u_2}=\sqrt {\tau _2}[\ket{C;{\alpha _2,\beta _2}}-\lambda ({\alpha _1,\beta _1};{\alpha _2,\beta _2})\ket{C;{\alpha _1,\beta _1}}]\nonumber\\
&& \bra{C;{\alpha _1,\beta _1}}u_2\rangle=0
\end{eqnarray}
This can be written in a compact way as
\begin{eqnarray}\label{gg}
&&\Pi ({\alpha _1,\beta _1};{\alpha _2,\beta _2})=\Pi({\alpha _1,\beta _1})+\varpi ({\alpha _2,\beta _2}|{\alpha _1,\beta _1})\nonumber\\
&&\varpi ({\alpha _2,\beta _2}|{\alpha _1,\beta _1})=\ket{u_2}\bra{u_2}=
\frac{\Pi^{\perp}({\alpha _1,\beta _1})\Pi({\alpha _2,\beta _2})\Pi^{\perp}({\alpha _1,\beta _1})}{{\rm Tr}[\Pi^{\perp}({\alpha _1,\beta _1})\Pi({\alpha _2,\beta _2})]}\nonumber\\
&&\Pi ({\alpha _1,\beta _1})\varpi ({\alpha _2,\beta _2}|{\alpha _1,\beta _1})=0.
\end{eqnarray}

We generalize this as follows. From the $d^2$ coherent states, we consider $d$ linearly independent coherent states
(with appropriate choice of the fiducial vector, we can make any set of $d$ coherent states that we choose, linearly independent).
We order them in an arbitrary way and label them as $\ket {C;\alpha _1, \beta _1}$,...,$\ket {C;\alpha _n, \beta _n}$. 
Using the Gram-Schmidt orthogonalization algorithm, we express the projector $\Pi ({\alpha _1,\beta _1};...;{\alpha _i,\beta _i})$ (where $i=2,...,d$), as
\begin{eqnarray}\label{alg10}
&&\Pi({\alpha _1,\beta _1};...;{\alpha _i,\beta _i})=\Pi({\alpha _1,\beta _1};...;{\alpha _{i-1},\beta _{i-1}})+\ket{u_i}\bra{u_i}\nonumber\\
&&\ket{u_i}=\sqrt {\tau _i}\Pi ^{\perp}({\alpha _1,\beta _1};...;{\alpha _{i-1},\beta _{i-1}})\ket{{\alpha _i,\beta _i}}\nonumber\\
&&\tau_i=\{{\rm Tr}[\Pi^{\perp}({\alpha _1,\beta _1};...;{\alpha _{i-1},\beta _{i-1}})\Pi({\alpha _i,\beta _i})]\}^{-1}
\end{eqnarray}
where $\Pi  ^{\perp} ({\alpha _1,\beta _1};...;{\alpha _{i-1},\beta _{i-1}})={\bf 1}-\Pi  ({\alpha _1,\beta _1};...;{\alpha _{i-1},\beta _{i-1}})$.
We rewrite this as
\begin{eqnarray}\label{alg1}
&&\Pi({\alpha _1,\beta _1};...;{\alpha _{i},\beta _{i}})=\Pi({\alpha _1,\beta _1};...;{\alpha _{i-1},\beta _{i-1}})+\varpi ({\alpha _{i},\beta _{i}}|{\alpha _1,\beta _1};...;{\alpha _{i-1},\beta _{i-1}})\nonumber\\
&&\varpi ({\alpha _{i},\beta _{i}} |{\alpha _1,\beta _1};...;{\alpha _{i-1},\beta _{i-1}})=\frac{\Pi  ^{\perp} ({\alpha _1,\beta _1};...;{\alpha _{i-1},\beta _{i-1}}) \Pi({\alpha _{i},\beta _{i}})
\Pi ^{\perp}({\alpha _1,\beta _1};...;{\alpha _{i-1},\beta _{i-1}})}
{{\rm Tr}[\Pi ^{\perp}({\alpha _1,\beta _1};...;{\alpha _{i-1},\beta _{i-1}})\Pi ({\alpha _{i-1},\beta _{i-1}})]}\nonumber\\
&&{\rm Tr}[\Pi ({\alpha _1,\beta _1};...;{\alpha _{i},\beta _{i}})]=i;\;\;\;\;\;{\rm Tr}[\varpi ({\alpha _{i-1},\beta _{i-1}}|{\alpha _1,\beta _1};...;{\alpha _{i-1},\beta _{i-1}})]=1,
\end{eqnarray}
The denominator ${\rm Tr}[\Pi ^{\perp}({\alpha _1,\beta _1};...;{\alpha _{i-1},\beta _{i-1}})\Pi ({\alpha _{i-1},\beta _{i-1}})]$ is non-zero, because we have considered linearly independent coherent states.
Related to the above algorithm is the $QR$ factorization of matrices \cite{matrix}, which is readily available in computer libraries (e.g., in MATLAB). 
It orthogonalizes the $k$ vectors (columns), in a $k\times k$ matrix.
It is clear that
\begin{eqnarray}\label{590}
\Pi({\alpha _1,\beta _1};...;{\alpha _{i},\beta _{i}})=\Pi({\alpha _1,\beta _1})+\varpi ({\alpha _2,\beta _2}|{\alpha _1,\beta _1})+...
+\varpi ({\alpha _{i},\beta _{i}}|{\alpha _1,\beta _1};...;{\alpha _{i-1},\beta _{i-1}}).
\end{eqnarray}
We call the $\Pi({\alpha _1,\beta _1};...;{\alpha _{i},\beta _{i}})$, $\varpi ({\alpha _{i},\beta _{i}}|{\alpha _1,\beta _1};...;{\alpha _{i-1},\beta _{i-1}})$ coherent projectors, 
because of their resolution of the identity property, and the fact that they are transformed into operators of the same type under displacement transformations.
The following propositions discuss this.

\begin{proposition}\label{PR1}
Under displacement transformations the $\Pi({\alpha _1,\beta _1};...;{\alpha _{i},\beta _{i}})$, $\varpi({\alpha _{i},\beta _{i}}|{\alpha _1,\beta _1};...;{\alpha _{i-1},\beta _{i-1}})$,
${\mathfrak D}({\alpha _1,\beta _1};...;{\alpha _{i},\beta _{i}})$, are transformed into operators of the same type:
\begin{eqnarray}\label{vvv}
&&D({\kappa, \lambda})\Pi({\alpha _1,\beta _1};...;{\alpha _{i},\beta _{i}})[D({\kappa, \lambda})]^{\dagger}=
\Pi({\kappa+\alpha _1,\lambda+\beta _1};...;{\kappa+\alpha _{i},\lambda+\beta _{i}})\nonumber\\
&&D({\kappa, \lambda})\varpi({\alpha _{i},\beta _{i}}|{\alpha _1,\beta _1};...;{\alpha _{i-1},\beta _{i-1}})[D({\kappa, \lambda})]^{\dagger}=
\varpi({\kappa+\alpha _{i},\lambda+\beta _{i}}|{\kappa+\alpha _1,\lambda+\beta _1};...;{\kappa+\alpha _{i-1},\lambda+\beta _{i-1}})\nonumber\\
&&D({\kappa, \lambda}){\mathfrak D}({\alpha _1,\beta _1};...;{\alpha _{i},\beta _{i}})[D({\kappa, \lambda})]^{\dagger}=
{\mathfrak D}({\kappa+\alpha _1,\lambda+\beta _1};...;{\kappa+\alpha _{i},\lambda+\beta _{i}})
\end{eqnarray}
\end{proposition}
\begin{proof}
For  $\varpi({\alpha _{i},\beta _{i}}|{\alpha _1,\beta _1};...;{\alpha _{i-1},\beta _{i-1}})$, we prove the statement inductively using Eq.(\ref{alg1}) and the property of coherent states in Eq.(\ref{1111}).
Then we can use Eq.(\ref{590}) to prove the statement for $\Pi({\alpha _1,\beta _1};...;{\alpha _{i},\beta _{i}})$, and Eq.(\ref{106}) to prove the statement for ${\mathfrak D}({\alpha _1,\beta _1};...;{\alpha _{i},\beta _{i}})$. 
\end{proof}

\begin{proposition}\label{PR2}
\begin{itemize}
For fixed ${\alpha _1,\beta _1};...;{\alpha _{i},\beta _{i}}$, and $i=2,...,d$:
\item[(1)]
The following resolution of the identity holds:
\begin{eqnarray}\label{n1}
\frac{1}{id}\sum _{\kappa, \lambda}\Pi({\kappa+\alpha _1,\lambda+\beta _1};...;{\kappa+\alpha _{i},\lambda+\beta _{i}})={\bf 1}.
\end{eqnarray}
\item[(2)]
The following resolution of the identity holds:
\begin{eqnarray}\label{n2}
\frac{1}{i}\sum _{\kappa, \lambda}\varpi({\kappa+\alpha _{i},\lambda+\beta _{i}}|{\kappa+\alpha _1,\lambda+\beta _1};...;{\kappa+\alpha _{i-1},\lambda+\beta _{i-1}})={\bf 1}.
\end{eqnarray}
\item[(3)]
\begin{eqnarray}\label{n3}
\sum _{\kappa, \lambda}{\mathfrak D}({\kappa+\alpha _1,\lambda+\beta _1};...;{\kappa+\alpha _{i},\lambda+\beta _{i}})=0.
\end{eqnarray}
\end{itemize}
\end{proposition}
\begin{proof}
\mbox{}
\begin{itemize}

\item[(1)]
For any operator $\Theta$ (see Eq.(119) in \cite{vour2}): 
\begin{eqnarray}\label{35}
\frac{1}{d}\sum_{\kappa, \lambda}D({\kappa, \lambda})\Theta[D({\kappa, \lambda})]^{\dagger}={\bf 1}{\rm Tr}(\Theta).
\end{eqnarray}
We use this relation with $\Pi({\alpha _1,\beta _1};...;\alpha _{i},\beta _{i})$
(in which case ${\rm Tr}(\Theta)=i$). Taking into account the first of Eqs.(\ref{vvv}), we prove Eq.(\ref{n1}).
\item[(2)]
We use Eq.(\ref{35}) with $\varpi({\alpha _{i},\beta _{i}}|{\alpha _1,\beta _1};...;{\alpha _{i-1},\beta _{i-1}})$
(in which case ${\rm Tr}(\Theta)=1$). Taking into account the second of Eqs.(\ref{vvv}), we prove Eq.(\ref{n2}).
\item[(3)]
We use Eq.(\ref{35}) with ${\mathfrak D}({\alpha _1,\beta _1};...;{\alpha _{i},\beta _{i}})$
(in which case ${\rm Tr}(\Theta)=0$). Taking into account the third of Eqs.(\ref{vvv}), we prove Eq.(\ref{n3}).

\end{itemize}
\end{proof}
The 
\begin{eqnarray}\label{mixed}
\rho=\frac{1}{n}\Pi({\alpha _1,\beta _1};...;\alpha _{n},\beta _{n})
\end{eqnarray}
is a density matrix of a mixed state (its entropy is $-{\rm Tr}(\rho \log \rho)=\log n$).
These density matrices obey the properties in propositions \ref{PR1},\ref{PR2}, and in this sense they are coherent mixed states, the practical importance of which requires further study.

\section{Discussion}

Kolmogorov probability theory is intimately connected to set theory, which is a Boolean algebra.
Quantum probabilities are non-additive, due to non-commutativity, and we have introduced the M\"obius (or non-additivity) operators
${\mathfrak D}(H_1, H_2)$, ${\mathfrak D}(H_1, H_2,H_3)$, etc, which are related to commutators as described in Eqs(\ref{e3}),(\ref{bg}).

Finite quantum systems are described with 
the Birkhoff-von Neumann  modular orthocomplemented lattice of subspaces, $\Lambda _d$. 
Unlike Boolean algebras, this is not a distributive lattice. We
quantified the lack of distributivity, with the projectors 
$\varpi _1(H_1,H_2|H_0)$ and $\varpi _2(H_1,H_2|H_0)$ in Eq.(\ref{40}).
The lack of distributivity is linked to the violation of the law of the total probability, which is fundamental for Kolmogorov probabilities.
The projectors $\pi(H_0;H_1)$ quantify deviations from the law of the total probability.

The operators ${\mathfrak D}(H_1, H_2)$, ${\mathfrak D}(H_1, H_2,H_3)$, $\varpi _1(H_1,H_2|H_0)$, $\varpi _2(H_1,H_2|H_0)$, $\pi(H_0;H_1)$ 
are Hermitian operators and can be measured experimentally.
We have calculated the average value and standard deviation of these quantities, for various examples.
If quantum probabilities were Kolmogorov probabilities, all these quantities would have been zero.

There are constraints on the projectors corresponding to quantum probabilities, imposed by the fact that $\Lambda _d$ is a modular lattice. 
They have been studied in propositions \ref{P1},\ref{P2},\ref{P3}.

The general theory has been used in the context of coherent states. This led to the projectors $\Pi({\alpha _1,\beta _1};...;{\alpha _{i},\beta _{i}})$,
$\varpi({\alpha _{i},\beta _{i}}|{\alpha _1,\beta _1};...;{\alpha _{i-1},\beta _{i-1}})$,
${\mathfrak D}({\alpha _1,\beta _1};...;{\alpha _{i},\beta _{i}})$, which generalize coherence to multi-dimensional structures. The term coherence is used here in the sense of propositions \ref{PR1},\ref{PR2}.
The practical importance of the coherent mixed states in Eq.(\ref{mixed}), requires further study.

The work provides insight to the nature of quantum probabilities and their relationship to the geometry of quantum mechanics described by the Birkhoff-von Neumann lattice.
In particular, it links the formalism of non-additive probabilities, with the non-commutativity formalism.

\end{document}